\documentclass[11pt]{article}
\usepackage{vmargin}
\setpapersize{USletter}
\setmargrb{0.9in}{0.8in}{0.9in}{0.8in}
\usepackage{amsmath}
\usepackage{amssymb}
\usepackage{amsthm}
\usepackage{ifthen}
\usepackage[T1]{fontenc}
\usepackage{cite}
\usepackage{psfrag}
\usepackage{graphics,graphicx}
\usepackage{algorithmic}
\usepackage{algorithm}
\usepackage{subfigure}

\newtheorem{theorem}{Theorem}
\newtheorem{lemma}[theorem]{Lemma}
\newtheorem{corollary}[theorem]{Corollary}
\newtheorem{proposition}[theorem]{Proposition}

\newtheorem{definition}[theorem]{Definition}

\newcommand{\runtitle}[1]{\vspace{1ex}{\small \textbf{\boldmath #1}}}
\newcommand{\bl}{\text{BL}}
\newcommand{\PBMP}{\textup{P-BMP}}
\newcommand{\OPBMP}{\textup{1D-P-BMP}}
\newcommand{\BMP}{\textup{BMP}}
\newcommand{\OBMP}{\textup{1D-BMP}}
\newcommand{\WMSA}{\textup{WMSA}}
\newcommand{\MRCT}{\textup{MRCT}}

\newcommand{\len}{{\ell}}
\newcommand{\probeset}{{\mathcal S}}
\newcommand{\neighbor}{{\mathcal N}}

\newcommand{\embed}{{\varepsilon}}
\newcommand{\placement}{{\phi}}
\newcommand{\borderlen}{\textup{BL}}
\newcommand{\border}{\textup{border}}
\newcommand{\share}{\textup{share}}
\newcommand{\mask}{{\mathcal M}}

\newcommand{\comment}[1]{}

\title{Hardness and Approximation of
\\ The Asynchronous Border Minimization Problem}
\author{Alexandru Popa\thanks{
Department of Computer Science, University of Bristol,
Email: \tt{popa@cs.bris.ac.uk}}
\and Prudence W.H. Wong\thanks{
Department of Computer Science, University of Liverpool,
Email: \tt{pwong@liverpool.ac.uk, ccyung@graduate.hku.hk}}
\and Fencol C.C. Yung$^\dagger$
}

\begin{document}

\begin{titlepage}

\maketitle
\thispagestyle{empty}

\comment{
Choices of title
\begin{itemize}
\item
  Hardness and Improved Approximation Results for
  The Border Minimization Problem
\item
  Complexity and Improved Approximation Results for
  The Asynchronous Border Minimization Problem
\item
  Approximation and Hardness Results for
  The Asynchronous Border Minimization Problem
\item NP-Hardness of Asynchronous Border Minimization Problem
  and Improved Approximation Algorithms
\end{itemize}
}

\begin{abstract}
We study a combinatorial problem arising from microarrays synthesis.
The synthesis is done by a light-directed chemical process.
The objective is to minimize unintended illumination that
may contaminate the quality of experiments.
Unintended illumination is measured by a notion called border length
and the problem is called Border Minimization Problem ($\BMP$).
The objective of the $\BMP$ is to place a set of probe sequences
in the array and find an embedding (deposition of nucleotides/residues to
the array cells) such that the sum of border length is minimized.
A variant of the problem, called $\PBMP$, is that the placement is given and
the concern is simply to find the embedding.

Approximation algorithms have been proposed for the problem~\cite{LWY+08}
but it is unknown whether the problem is NP-hard or not.
In this paper, we give a thorough study of different variations of $\BMP$
by giving NP-hardness proofs and improved approximation algorithms.
We show that $\PBMP$, $\OBMP$, and $\BMP$ are all NP-hard.
Contrast with the result in~\cite{LWY+08} that $\OPBMP$ is polynomial time solvable,
the interesting implications include
(i) the array dimension (1D or 2D) differentiates the complexity of $\PBMP$;
(ii) for 1D array, whether placement is given differentiates
the complexity of $\BMP$;
(iii) $\BMP$ is NP-hard regardless of the dimension of the array.
Another contribution of the paper is improving the approximation for $\BMP$
from $O(n^{1/2} \log^2 n)$ to $O(n^{1/4} \log^2 n)$, where $n$ is the total number of sequences.
\end{abstract}

\end{titlepage}

\section{Introduction}
\label{sec:intro}

DNA and peptide microarrays~\cite{GR+99,CC09} are important research tools
used in gene discovery, multi-virus discovery,
disease and cancer diagnosis.
Apart from measuring the amount of gene expression~\cite{ST+00},
microarray is an efficient tool for making a qualitative
statement about the presence or absence of biological target sequences
in a sample,
for example
peptide microarrays have been used for detecting tumor biomarkers~\cite{CMI+06,MES+04,WSK+03}.
A microarray is a plastic or glass slide (2D grid-like) consisting of thousands of
short sequences called \emph{probes}.
Microarray design raises a number of challenging combinatorial problems, such as probe
selection~\cite{GL+07,KS02,LS01,SL03},
deposition sequence design~\cite{KZ+02,R03} and probe placement
and synthesis~\cite{HH+02,CR06,CR061,CR07,KM+04,KM+06}.
In this paper, we focus on the probe placement and synthesis problem.

The synthesis process~\cite{FR+91} consists of two components:
\emph{probe placement} and \emph{probe embedding}.
Probe placement is to place each probe to a unique array cell
and probe embedding is to determine a \emph{deposition sequence}
of masks 
to allow (and block) lights on the array cells
(see Figure~\ref{fig:mask}).
The deposition sequence is a supersequence of all probes.
Figure~\ref{fig:embedding} shows
the deposition sequence (ACGT)$^3$
and various embeddings of the probe CGT,
e.g., (a) shows the embedding ($-$)C($-$)$^4$G($-$)$^4$T,
where ``$-$'' represents a space.
The synthesis can be classified as
\emph{synchronous} and \emph{asynchronous} synthesis.
In the former, the $i$-th deposition character can only be
deposited to the $i$-th position of the probes.
In the later, there is no such restriction.
Figure~\ref{fig:mask} shows asynchronous synthesis.

Due to diffraction, internal reflection and scattering,
cells on the \emph{border} between masked and unmasked regions
are often subject to unintended illumination~\cite{FR+91},
and can compromise experimental results.
As microarray chip is expensive to synthesize,
unintended illumination should be minimized.
The magnitude of unintended illumination can be measured by the
\emph{border length} of the masks used,
which is the number of borders shared between masked and
unmasked regions,
e.g., in Figure~\ref{fig:mask}, the border
length of $\mask_1,\mask_3,\mask_4$ is $2$ and $\mask_2$ is~$4$.

\runtitle{Synchronous synthesis.}
Hanannenhalli et al.~\cite{HH+02} defined the
\emph{Border Minimization Problem} ($\BMP$)
for synchronous synthesis,
in which the only concern is probe placement.
Once the placement is fixed,
the border length is proportional to the Hamming distance
of neighboring probes.
Hanannenhalli et al.~\cite{HH+02} proposed an approximation algorithm
based on travelling salesman path (TSP) in the complete graph
representing the probes and their Hamming distance.
Experiments have been carried out to show the effectiveness of the algorithm.
Other algorithms~\cite{KM+04,KM+06,CR061} have been proposed to
improve the experimental results.
Recently, the problem has been proved to be NP-hard~\cite{KR09}
and $O(\sqrt{n})$-approximable~\cite{KRD10},
where $n$ is the number of probes.

\runtitle{Asynchronous synthesis.}
In this paper, we focus on
asynchronous synthesis,
which was introduced by Kahng et al.~\cite{KM+04}.
The problem appears to be difficult as
they studied a special case that
the deposition sequence is given and the embeddings of
all but one probes are known.
A polynomial time dynamic programming algorithm was proposed
to compute the optimal embedding of this single probe.
This algorithm is used as the basis for several
heuristics~\cite{CR06,CR061,CR07,KM+04,KM+06} that are
shown experimentally to reduce
unintended illumination. 
The 
dynamic programming mentioned
above computes the optimal embedding of a single probe
in time $O(\len |D|)$,
where $\len$ is the length of a probe and
$D$ is the deposition sequence.
The algorithm can be
extended to an exponential time algorithm to find the optimal
embedding of all $n$ probes
in $O(2^n \len^n|D|)$ time.

\begin{figure}[t]
  \begin{center}
 {
  \psfrag{e1}{\footnotesize{$e_1 = $}}
  \psfrag{e2}{\footnotesize{$e_2 = $}}
  \psfrag{e3}{\footnotesize{$e_3 = $}}
  \psfrag{e4}{\footnotesize{$e_4 = $}}
  \psfrag{e5}{\footnotesize{$e_5 = $}}
  \psfrag{M1}{\scriptsize{$\mask_1$}}
  \psfrag{M2}{\scriptsize{$\mask_2$}}
  \psfrag{M3}{\scriptsize{$\mask_3$}}
  \psfrag{M4}{\scriptsize{$\mask_4$}}
  \psfrag{M5}{\scriptsize{$\mask_5$}}
  \psfrag{A}{\tiny{A}}
  \psfrag{C}{\tiny{C}}
  \psfrag{G}{\tiny{G}}
  \psfrag{T}{\tiny{T}}
  \psfrag{AC}{\tiny{AC}}
  \psfrag{CA}{\tiny{CA}}
  \psfrag{CT}{\tiny{CT}}
  \psfrag{TA}{\tiny{TA}}
  \psfrag{unmasked region}{\scriptsize{unmasked region}}
  \psfrag{masked region}{\scriptsize{masked region}}

  \includegraphics[width=8cm]{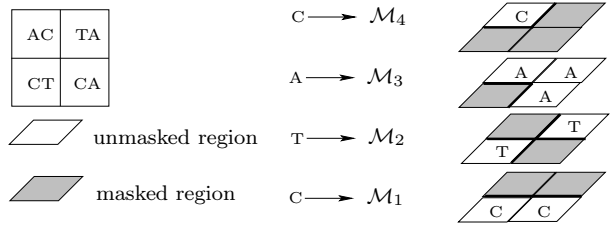}
 }
   \caption{Asynchronous synthesis of a $2 \times 2$ microarray.
The deposition sequence $D =$ CTAC
corresponds to the sequence of four masks $\mask_1$, $\mask_2$, $\mask_3$,
and $\mask_4$.
The masked regions are shaded.
The borders between the masked and unmasked regions are
represented by bold lines.
With respect to the deposition sequence $D$,
the embedding of sequence AC is $--$AC, TA is $-$TA$-$,
CT is CT$--$, CA is C$-$A$-$.}
  \label{fig:mask}
  \end{center}
  \end{figure}

\begin{figure}[t]
  \begin{center}
 {
  \psfrag{p}{\footnotesize{$p =$ CGT}}
  \psfrag{1}{\footnotesize{(a)}}
  \psfrag{2}{\footnotesize{(b)}}
  \psfrag{3}{\footnotesize{(c)}}
  \psfrag{4}{\footnotesize{(d)}}
  \psfrag{a}{\scriptsize{A}}
  \psfrag{c}{\scriptsize{C}}
  \psfrag{g}{\scriptsize{G}}
  \psfrag{t}{\scriptsize{T}}
  \psfrag{S}{{$D$}}
  \psfrag{E1}{{$\embed_1$}}
  \psfrag{E2}{{$\embed_2$}}
  \psfrag{E3}{{$\embed_3$}}
  \psfrag{E4}{{$\embed_4$}}
  \includegraphics[width=6cm]{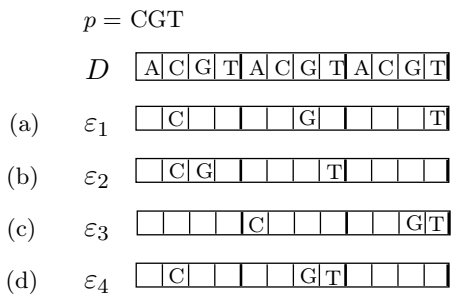}
 }
   \caption{Different embeddings of probe $p =$ CGT into deposition sequence
$D = (\mbox{ACGT})^3$.}
  \label{fig:embedding}
  \end{center}
  \end{figure}

To find both placement and embedding,
Li et al.~\cite{LWY+08} proposed
an $O(\sqrt{n} \log^2 n)$-approximation algorithm.
This is based on their $O(\log^2 n)$-approximation 
for finding embeddings when placement is given.
On a one-dimensional array,
they showed that the approximation ratio could be improved to $3/2$.
If in addition the placement is given, the problem can be solved
optimally in polynomial time.
It is however unknown whether the general problem is NP-hard or not.
We note that the NP-hard proof in~\cite{KR09} cannot be applied
to the asynchronous problem and it is not straightforward to show
direct relation between the synchronous and asynchronous problems.
This leaves several open questions in the asynchronous problem.
Let us denote by $\PBMP$ the problem with placement already given.

\begin{itemize}
\item So far only approximation algorithms for $\BMP$ have been proposed.
  An open question is whether $\BMP$ (or even $\OBMP$) is NP-hard.
\item An exponential-time algorithm for $\PBMP$ has been proposed
  while $\OPBMP$ can be solved optimally in polynomial time.
  Is $\PBMP$ (on 2D array) NP-hard?
\item Is it possible to improve existing approximation algorithms
  for $\BMP$ or $\PBMP$?
\end{itemize}

\begin{table}[t]
\begin{center}
\begin{tabular}{|c|c|c|c|c|} \hline
Setting & 2D & 1D 
\\ \hline \hline
$\BMP$ & NP-hard$^*$ & NP-hard$^*$ 
\\
 & $O(n^{1/4}\log^2 n)$-approximate$^*$ & $\frac{3}{2}$-approximate~\cite{LWY+08} 
\\ \hline
$\PBMP$ & NP-hard$^*$ & polynomial time solvable~\cite{LWY+08} 
\\
& $O(\log n)$-approximate$^*$ & 
\\ \hline
\end{tabular}
\caption{Results on $\BMP$ and $\PBMP$.
  Results in this paper are marked with an asterisk.}
\label{tab:results}
\end{center}
\end{table}

\runtitle{Our contributions.}
We give a thorough study of different variations
of the asynchronous border minimization problem.
We answer the above questions affirmatively
by giving several NP-hard proofs and better approximation algorithms.
Our contributions are listed below (see Table~\ref{tab:results} also):

\begin{itemize}
\item For $\OBMP$ (placement not given), 
  we give a reduction from the Hamiltonian Path problem~\cite{GJ90}
  to $\OBMP$, implying the NP-hardness of $\OBMP$. \\
  This means that when the array is one dimension,
  whether placement is given differentiates the complexity of $\BMP$
  (as $\OPBMP$ is polynomial time solvable~\cite{LWY+08}).
\item We further show that $\OBMP$ can be reduced to $\BMP$
  and thus, $\BMP$ is NP-hard. \\
  This means that $\BMP$ is NP-hard regardless of the dimension of the array.
\item For $\PBMP$, 
  we show that the Shortest Common Supersequence problem~\cite{RU81} can be
  reduced to $\PBMP$,
  implying that $\PBMP$ is NP-hard. \\
  This means that the dimension differentiates the complexity of $\PBMP$
  as we have seen in~\cite{LWY+08} that
  $\OPBMP$ is polynomial time solvable.
\item We also improve the approximation ratio for $\PBMP$ from $O(\log^2 n)$
  to $O(\log n)$ and $\BMP$ from $O(n^{1/2}\log^2 n)$ to $O(n^{1/4}\log^2 n)$.
\end{itemize}

\runtitle{Organization of the paper.}
In Section~\ref{sec:prelim}, we give some definitions
and preliminaries.
In Sections~\ref{sec:pbmp} and~\ref{sec:bmp},
we give the hardness results for $\PBMP$ and $\BMP$, respectively.
In Section~\ref{sec:approx},
we present and analyze an approximation algorithm for $\BMP$.
We conclude in Section~\ref{sec:conclude}.

\section{Preliminaries}
\label{sec:prelim}
We are given a set of $n$ sequences 
$\probeset = \{s_1, s_2,\ldots, s_n\}$,
a $\sqrt{n} \times \sqrt{n}$ array
(for simplicity, we assume that $\sqrt{n}$ is an integer).
For any sequence $s_i$,
we denote the length of the sequence by $\ell_i$ and
the $t$-th character of a sequence $s_i$ by $s_i[t]$.
The probe sequences in $\probeset$ are to be placed on the
$\sqrt{n} \times \sqrt{n}$ array.
We denote a cell in the array as $v$.
Two cells $v_1=(x_1, y_1)$ and $v_2=(x_2, y_2)$
are said to be {\em neighbors}
if $|x_1 - x_2| + |y_1 - y_2| = 1$.
For each cell $v$, we denote the
set of neighbors of $v$ by $\neighbor(v)$.

\runtitle{Placement and embedding.}
A {\em placement} of the probe sequences is a
bijective function $\placement$
that maps each probe sequence
to a unique cell in the array.
A deposition sequence $D$ is a sequence of characters
and each character is deposited by a mask
to some cells of the array.
A mask $\mask$ can be viewed as a 2D-array
such that $\mask(i,j)$ is either the character
associated with $\mask$ or a space ``$-$''.
The space means that the character is not deposited
in this cell.

An {\em embedding} of a set of probes $\probeset$ into
a deposition sequence $D$ is
denoted by $\embed = \{\embed_1, \embed_2, \ldots, \embed_n \}$.
For $1\leq i \leq n$,
$\embed_i$ is a length-$|D|$ sequence such that
(1) $\embed_i[t]$ is either $D[t]$ or a space $``-"$;
and (2) removing all spaces from $\embed_i$ gives $s_i$.
There are two ways to define the border length between
two probes $s_i$ and $s_j$.
The Hamming distance between $\embed_i$ and $\embed_j$
measures $\border_{\embed}(s_i,s_j)$.
With this definition, $\border_{\embed}(s_i,s_j) = \border_{\embed}(s_j,s_i)$.
We use this in Section~\ref{sec:approx}.
Sometimes, it is more convenient to define an asymmetric measurement,
$\border'_{\embed}(s_i,s_j)$
is the number of $p$'s such that
(i) $\embed_i[p] \not= \text{`$-$'}$, and
(ii) $\embed_i[p] \not= \embed_j[p]$.
Condition (ii) means that $\embed_j[p] = \text{`$-$'}$.
Note that $\border'_{\embed}(s_i,s_j) \not = \border'_{\embed}(s_j,s_i)$
while $\border_{\embed}(s_i,s_j) = \border'_{\embed}(s_i,s_j) + \border'_{\embed}(s_j,s_i)$.
We use this definition in Sections~\ref{sec:pbmp} and~\ref{sec:bmp}.

\runtitle{Border length.}
The {\em border length} of a placement $\placement$
and an embedding $\embed$ is defined as
the sum of borders over all pairs of probe sequences
\begin{equation}
 \borderlen(\placement, \embed)=
   \frac{1}{2}
   \displaystyle \sum_{\scriptsize
    \begin{array}{c}
    s_i, s_j:\\
    \placement(s_j) \in \neighbor(\placement(s_i))
    \end{array}
    } \border_{\embed}(s_i, s_j)
 = \sum_{\scriptsize
    \begin{array}{c}
    s_i, s_j:\\
    \placement(s_j) \in \neighbor(\placement(s_i))
    \end{array}
    } \border'_{\embed}(s_i, s_j)
    \, .
 \label{eq:bmpcost}
\end{equation}

We can also define border length in terms of the border length of all the masks.
For any mask $\mask$ of deposition character $X$,
the border length of $\mask$, denoted by $\borderlen(\mask)$,
is defined as
the number of neighboring cells $(i_1,j_1)$ and $(i_2,j_2)$
such that $\mask(i_1,j_1)=X$ and $\mask(i_1,j_1) \not= \mask(i_2,j_2)$.
For a placement and embedding that corresponds to
a sequence of masks $\mask_1$, $\mask_2$, $\cdots$, $\mask_d$,
\begin{equation}
 \borderlen(\placement, \embed)=
  \sum_{h=1}^{d} \borderlen(\mask_h)
 \label{eq:maskcost}
\end{equation}

The objective is to find a placement $\placement$ and
an embedding~$\embed$, so that $\borderlen(\placement, \embed)$
is minimized.

When the placement is given, we call the problem the $\PBMP$.
We also consider the $\BMP$ when the array is one dimensional
and we call the problem $\OBMP$.

\runtitle{$\WMSA$ and $\MRCT$.}
As shown in~\cite{LWY+08}, $\PBMP$ can be reduced to
the \emph{weighted multiple sequence alignment problem} ($\WMSA$),
which in turn can be reduced to the \emph{minimum routing cost tree problem} ($\MRCT$).
In the $\WMSA$ problem~\cite{BV01,FD87,G93,RL+97},
we are given $k$ sequences $S=\{S_1, S_2, \cdots, S_k\}$.
An alignment is $S'=\{S'_1, S'_2, \cdots, S'_k\}$
such that all $S'_i$ have the same length
and $S'_i$ is formed by inserting spaces into $S$.
The problem is to minimize the \emph{weighted sum-of-pair score} 
which is the weighted sum of pair-wise score of every pair of sequences
in the alignment and the pair-wise score is the sum of distance of characters
in each position of the two sequences.
In the $\MRCT$ problem~\cite{Bartal96},
we are given a graph with weighted edges.
In a spanning tree,
the routing cost between two vertices is the sum of weights of the edges
on the unique path between the two vertices in the spanning tree.
The $\MRCT$ problem is to find a spanning tree with minimum routing cost,
which is defined as the sum of routing cost
between every pair of two vertices.

The reduction results in~\cite{LWY+08} imply the following lemma.
\begin{lemma}[{\!\!\cite{LWY+08}}]
\label{thm:mrct}
If there is a $c$-approximation for $\MRCT$,
there is a $c$-approximation for $\PBMP$.
\end{lemma}
It is stated in~\cite{LWY+08} that
Bartal's algorithm~\cite{Bartal96} finds a routing spanning tree
by embedding a metric space into a distribution of trees with expected distortion $O(\log^2 n)$.
$\MRCT$ is $O(\log^2 n)$-approximable~\cite{Bartal96}.
Meanwhile, the ratio is improved to $O(\log n)$ by Fakcharoenphol, Rao and Talwar~\cite{FRT03}.
Together with Lemma~\ref{thm:mrct}, we have the following corollary.
\begin{corollary}
\label{thm:pbmp}
There is an $O(\log n)$-approximation for the $\PBMP$.
\end{corollary}
Notice that we use the term embedding in two contexts,
probe embedding refers to finding the deposition sequence
while embedding a metric to trees is to obtain an approximation.
This should be clear from the context and should not cause confusion.

\section{$\PBMP$: Finding embedding when placement is given}
\label{sec:pbmp}

We give a reduction from
the Shortest Common Supersequence (SCS) problem
to the $\PBMP$.

\runtitle{Shortest Common Supersequence problem.}
Given $n$ sequences of characters,
a common supersequence is a sequence that contains
all the $n$ sequences as subsequences.
The Shortest Common Supersequence problem is to
find a common supersequence with the minimum length.

The reduction is from the SCS problem over binary alphabet,
which is known to be NP-hard~\cite{RU81}.
Suppose that the binary alphabet is $\{0,1\}$.
Consider an instance of the SCS problem
with a set $S$ of $k$ binary strings $s_1, \cdots, s_k$.
Let $\ell_i$ be the length of $s_i$,
$\ell = \max_{1 \leq i \leq k} \ell_i$
be the length of the longest sequence in $S$,
and $L = \sum_{1 \leq i \leq k} \ell_i$.
For any $1 \leq p, q \leq \ell$,
we define an instance for $\PBMP$,
denoted by $I(p,q)$.
As we show later,
a shortest common supersequence can be found by
computing the optimal solutions for a polynomial number
of instances $I(p,q)$.

\runtitle{The input $I(p,q)$.}
We construct a $(2k+1) \times (2k+1)$ array.
The probe sequences are over the alphabet $\{0,1,\$\}$,
where $\$$ is a character different from 0 or 1.

\begin{itemize}
\item Except for row 2-4, each cell of row 1, 5, 6, 7, 8, $\cdots$
  of the array contains the string ``\$''.
  We call these rows \emph{dummy-rows}.
\item All the cells of row 2 contain the same string ``$0^p$''.
  We call this row \emph{all-0-row}.
\item All the cells of row 4 contain the same string ``$1^q$''.
  We call this row \emph{all-1-row}.
\item Row 3 contains $s_1$, $s_2$, $\cdots$, $s_k$ in alternate cells,
  and the rest of the cells contain the string ``\$'',
  precisely, row 3 contains ``\$'', $s_1$, ``\$'', $s_2$, ``\$'', $\cdots$, ``\$'', $s_k$, ``\$''.
  We call this row \emph{seq-row}.
\end{itemize}

Tables~\ref{tab:cs} and~\ref{tab:scs} show
examples of $I(3,3)$
and $I(1,1)$, respectively.

\runtitle{Common supersequence and deposition sequence.}
Consider an instance $I(p,q)$,
we need at least one mask for the dummy strings ``\$'',
and the best is to use exactly one mask, say $M_{\$}$ for all these strings.
For $M_{\$}$,
row 1 (dummy-row) incurs a border length of $2k+1$ on the bottom boundary with all-0-row,
and row 5 (dummy-row) incurs $2k+1$ on the top boundary with all-1-row.
For seq-row, the border length on top boundary with all-0-row is $k+1$,
on bottom boundary with all-1-row is also $k+1$,
and within seq-row on left and right boundaries is $2k$.
Therefore,
the border length $\bl(M_{\$}) = 4(2k+1)$.
The total border length for $I(p,q)$
equals to $\bl(M_{\$})$ plus that of the remaining deposition sequence,
which in turn is related to a common supersequence
of the sequences in $S$.
Since the quantity $\bl(M_{\$})$ is present in all the embeddings,
we ignore this quantity when we discuss the border length for $I(p,q)$.
The following lemma states a relationship between
a common supersequence and an embedding of the probe sequences.
Table~\ref{tab:cs} gives an example.

\begin{table}[t]
\begin{center}
\begin{tabular}{|c|c|c|c|c|c|c|}
\hline
 \quad \$\quad\,  & \quad \$\quad\,  & \quad \$\quad\,  & \quad \$\quad\,  & \quad \$\quad\,  & \quad \$\quad\,  & \quad \$\quad\,  \\
\hline
 000 & 000 & 000 & 000 & 000 & 000 & 000 \\
\hline
 \$ & 010 & \$ & 100 & \$ & 00 & \$ \\
\hline
 111 & 111 & 111 & 111 & 111 & 111 & 111 \\
\hline
 \$ & \$ & \$ & \$ & \$ & \$ & \$ \\
\hline
 \$ & \$ & \$ & \$ & \$ & \$ & \$ \\
\hline
\$ & \$ & \$ & \$ & \$ & \$ & \$ \\
\hline
\end{tabular}
\end{center}
\caption{$s_1 = $ ``010'', $s_2 = $ ``100'', $s_3 = $``00''.
The supersequence $D=$ ``010011'' 
is an optimal deposition sequence for $I(3,3)$.
Ignoring the border of the mask for the dummy strings ``\$'',
the optimal border length equals
$2(p^*+q^*) \times (2k + 1) + 2L = 100$,
where $p^*=q^*=k=3$ and $L=8$.
}
\label{tab:cs}
\end{table}

\begin{lemma}
\label{thm:superseqDepositPos}
If $D$ is a common supersequence of the sequences in $S$
and the number of $0$'s and $1$'s in $D$ is $p^*$ and $q^*$, respectively,
then $D$ is an optimal deposition sequence for $I(p^*,q^*)$
and the resulting optimal embedding has a border length of
$2(p^*+q^*)(2k+1) + 2L$.
\end{lemma}

\begin{proof}
First of all, it is easy to observe that $D$ is a deposition sequence
for $I(p^*,q^*)$ because it is a common supersequence of sequences in $S$
and it has the same number of $0$'s and $1$'s in all-0-row and all-1-row
of the array in $I(p^*,q^*)$, respectively.
Notice that $p^*$ is at least the number of $0$'s in each of $s_i$
and similarly $q^*$ is at least the number of $1$'s.
In the deposition sequence $D$,
when $D[j]=0$, all-0-row incurs a border length of $2k+1$
on the top boundary with row 1 (dummy-row);
all-0-row and seq-row together incur a border length of $2k+1$
on the bottom boundary;
and a border length of $2x$ within seq-row,
where $x$ is the number of cells on seq-row that 0 is deposited.
A similar calculation can be done for the case when $D[j]=1$.
As a whole, the total border length equals
$2(p^*+q^*)(2k+1) + 2L$.

We further argue that this is the minimum border length for $I(p^*,q^*)$.
In any deposition sequence, the number of 0's is at least $p^*$
and the number of 1's is at least $q^*$.
Therefore,
all-0-row and
the cells with `0' on seq-row together incur a border length at least $2p^*(2k+1)$,
and similarly, all-1-row and
the cells with `1' on seq-row incur at least $2q^*(2k+1)$.
The cell on seq-row with the sequence~$s_i$ incurs $2\ell_i$ on
the left and right boundaries,
implying all these cells together incur $2L$.
Therefore, no matter how we deposit characters to the cell,
the total border length is at least $2(p^*+q^*)(2k+1)+2L$.
\end{proof}

Lemma~\ref{thm:superseqDepositPos} implies that
if $p+q$ is large enough,
we have a formula for the optimal border length
of the instance $I(p,q)$
in terms of $p$, $q$, and $L$.
The following lemma considers the situation when
$p+q$ is small.
Table~\ref{tab:scs} gives an example.

\begin{table}[t]
\begin{center}
\begin{tabular}{|c|c|c|c|c|c|c|}
\hline
 \quad \$\quad\, & \quad \$\quad\,  & \quad \$\quad\,  & \quad \$\quad\,  & \quad \$\quad\,  & \quad \$\quad\,  & \quad \$\quad\,  \\
\hline
 0 & 0 & 0 & 0 & 0 & 0 & 0 \\
\hline
 \$ & 010 & \$ & 100 & \$ & 00 & \$ \\
\hline
 1 & 1 & 1 & 1 & 1 & 1 & 1 \\
\hline
 \$ & \$ & \$ & \$ & \$ & \$ & \$ \\
\hline
 \$ & \$ & \$ & \$ & \$ & \$ & \$ \\
\hline
\$ & \$ & \$ & \$ & \$ & \$ & \$ \\
\hline
\end{tabular}
\end{center}
\caption{$s_1 = $ ``010'', $s_2 = $ ``100'', $s_3 = $``00''.
The shortest common supersequence is $D=$ ``0100''.
The optimal deposition for $I(1,1)$ is $D$.
Ignoring the border of the mask for the dummy strings ``\$'',
the optimal border length equals to
$(2\times 7 + 2\times 7 + 2 \times 3 + 2\times 2) + 2\times 8= 54$
(the first four terms refer to border length with top and bottom boundaries
and the last term with left and right boundaries).
On the other hand, $2(p^*+q^*) \times (2k + 1) + 2L = 44 < 54$,
where $p^*=q^*=1$, $k=3$ and $L=8$.
}
\label{tab:scs}
\end{table}

\begin{lemma}
\label{thm:superseqDepositNeg}
If $D$ is a shortest common supersequence of the sequences in $S$
and the number of $0$'s and $1$'s in $D$ is $p^*$ and $q^*$, respectively,
then for any $p_1, q_1$ such that
$p_1+q_1 < p^*+q^*$,
the optimal embedding for $I(p_1,q_1)$ has a border length
greater than
$2(p_1+q_1)(2k+1) + 2L$.
\end{lemma}

\begin{proof}
Notice that any deposition sequence must be a common supersequence,
and thus must have total length $\ell_D \geq p^*+q^* > p_1+q_1$.
With this deposition sequence,
the border length equals to $2\ell_D k + 2(p_1+q_1)(k+1) + 2L
> 2(p_1+q_1) k + 2(p_1+q_1)(k+1) + 2L
= 2(p_1+q_1)(2k+1)+2L$.
The term $2(p_1+q_1) k$ refers to the top and bottom border length
for columns with $s_i$ in the seq-row
while the term $2(p_1+q_1)(k+1)$ is for columns with dummy string ``\$''
in the seq-row.
\end{proof}

Using Lemmas~\ref{thm:superseqDepositPos} and~\ref{thm:superseqDepositNeg},
we can find the optimal solution for SCS from optimal solutions for $\PBMP$
as follows.
For all pairs of $1 \leq p \leq \ell$ and $1\leq q \leq \ell$,
we find the optimal solution to $I(p,q)$.
If the border length of the optimal solution equals to
$2(p+q)(2k+1) + 2L$,
there is a common supersequence of length $p+q$.
Among all such pairs of $p$ and $q$,
those with the minimum $p+q$ correspond to shortest common supersequences.
Notice that there are a polynomial number of, precisely $\ell^2$, pairs
of $p$ and $q$ to be checked.
We then have the following theorem.

\begin{theorem}
\label{thm:PBMP}
The $\PBMP$ is NP-Hard.
\end{theorem}

\section{$\BMP$: Finding placement and embedding}
\label{sec:bmp}

We first give a reduction from the Hamiltonian Path problem to $\OBMP$
(Section~\ref{sec:1dbmp})
and then a reduction from $\OBMP$ to $\BMP$ (Section~\ref{sec:2dbmp}).

\subsection{$\OBMP$: $\BMP$ on a 1D array}
\label{sec:1dbmp}

\runtitle{Hamiltonian Path problem.}
In an undirected graph a Hamiltonian Path is a path which visits each vertex exactly once.
Given an undirected graph, the problem is to decide if a Hamiltonian Path exists.
It is known that the problem is NP-hard~\cite{GJ90}.

\runtitle{Constructing an instance of $\OBMP$ from an instance of Hamiltonian Path problem.}
Consider an Hamiltonian Path instance in which
the given graph is $G=(V,E)$,
with $|V| = n$ and $|E| = m$.
Suppose that the vertices are labelled as $V = \{1, 2, \cdots, n \}$.
For any edge between vertex $i$ and $j$ with $i < j$,
we label the edge as $e_{i,j}$.
We construct an instance $I$ of $\OBMP$ of $n+2$ sequences to be placed
on an array of size $1 \times (n+2)$.
The size of the alphabet is $m+2$.
We now define the alphabet $\Sigma$ and the probe sequences $S$.
See Figure~\ref{fig:graph} and Table~\ref{tab:graph} for an example.

\begin{figure}[t]
  \begin{center}
 {
  \psfrag{v1}{\footnotesize{$v_1$}}
  \psfrag{v2}{\footnotesize{$v_2$}}
  \psfrag{v3}{\footnotesize{$v_3$}}
  \psfrag{v4}{\footnotesize{$v_4$}}
  \psfrag{v5}{\footnotesize{$v_5$}}
  \psfrag{v1e}{\footnotesize{$v_1:\quad e_{1,2}\, e_{1,5}$}}
  \psfrag{v2e}{\footnotesize{$v_2:\quad e_{1,2} \, e_{2,3} \, e_{2,4}$}}
  \psfrag{v3e}{\footnotesize{$v_3:\quad e_{2,3} \, e_{3,4}$}}
  \psfrag{v4e}{\footnotesize{$v_4:\quad e_{2,4} \, e_{3,4} \, e_{4,5}$}}
  \psfrag{v5e}{\footnotesize{$v_5:\quad e_{1,5} \, e_{4,5}$}}
  \includegraphics[height=2.5cm]{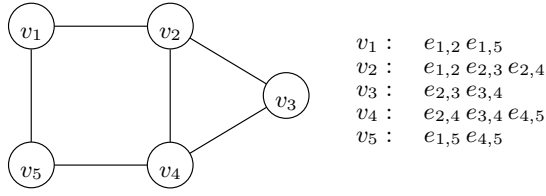}
 }
   \caption{A given graph $G$ of the Hamiltonian Path problem
    and the corresponding vertex sequences.}
  \label{fig:graph}
  \end{center}
\end{figure}

\begin{table}[t]
\begin{center}
\begin{tabular}{|c|c|c|c|c|c|c|}
\hline
\,\$\$\$\$\$\$\, &
\,$e_{1,2} \, e_{1,5}$\, &
\,$e_{1,2} \, e_{2,3} \, e_{2,4}$\, &
\,$e_{2,3} \, e_{3,4}$\, &
\,$e_{2,4} \, e_{3,4} \, e_{4,5}$\, &
$e_{1,5} \, e_{4,5}$ &
\,\#\#\#\#\#\#\,
\\ \hline
\multicolumn{7}{c}{}
\\ 
\multicolumn{7}{c}{}
\\ \hline
\$\$\$\$\$\$ &
$e_{1,2} \, e_{1,5}$ &
$e_{1,2} \, e_{2,3} \, e_{2,4}$ &
$e_{2,3} \, e_{3,4}$ &
$e_{1,5} \, e_{4,5}$ &
\,$e_{2,4} \, e_{3,4} \, e_{4,5}$ \, &
\#\#\#\#\#\#
 \\ \hline
\end{tabular}
\end{center}
\caption{Referring to the graph $G$ in Figure~\ref{fig:graph},
the placement on the top corresponds to vertex sequences
in the order of a Hamiltonian path $v_1, v_2, v_3, v_4, v_5$,
while the bottom one shows the order $v_1, v_2, v_3, v_5, v_4$,
which is not a Hamiltonian path.
For the placement on the top, we can deposit the edge characters
as $e_{1,2}, e_{2,3}, e_{2,4}, e_{3,4}, e_{1,5}, e_{4,5}$
and it can be seen that there is sharing for four edges
in the Hamiltonian path $e_{1,2}, e_{2,3}, e_{3,4}, e_{4,5}$.
For the placement on the bottom, we can deposit using the same sequence,
but there is sharing only for three edges,
namely, $e_{1,2}, e_{2,3}, e_{4,5}$.
}
\label{tab:graph}
\end{table}

\begin{enumerate}
\item Alphabet:
  $\Sigma = \{ e_{i,j} \mid e_{i,j} \in E \} \cup \{\$, \#\}$.
  We define an order on these characters such that
  $e_{i_1,j_1} < e_{i_2,j_2}$ if (i) $i_1 < i_2$ or (ii) $i_1=i_2$ and $j_1 < j_2$.

\item Probe sequences are divided into two types: vertex sequences
  and dummy sequences.

\begin{itemize}
\item Vertex sequences:
  For vertex $i$ in $V$,
  we construct a \emph{vertex sequence} $v_i$ such that
  $v_i$ is a string of the edge characters
  corresponding to all edges incident with it.
  The order of the edge characters in the string follows the same order defined in (1).
\item Dummy sequences:
  Furthermore, we add two \emph{dummy sequences} $w_{\$}$ and $w_{\#}$
  of length $n+1$ each, such that
  $w_{\$} = \$^{n+1}$, and $w_{\#} = \#^{n+1}$.
\end{itemize}
\end{enumerate}

Notice that the length of each vertex sequence $v_i$ is at most $n-1$.
Furthermore, for each vertex sequence, the order of the edge characters
follows the order defined in (1).
This implies that 
there exists a permutation of the edge characters such that
it is a common supersequence of all vertex sequences and
this forms a valid deposition sequence for all the vertex sequences.

The two dummy sequences $w_{\$}$ and $w_{\#}$
are to ensure that in an optimal placement
they will be placed in the leftmost and the rightmost cells
in the 1D array,
otherwise, the border length would be too large.
In other words, all the other vertex sequences $v_i$
will be placed in cells such that both left and right boundaries count.

We now claim that the graph $G$ has a Hamiltonian Path if and only if
the optimal border length of the $\OBMP$ instance $I$ is
$2(n+1) + (4m-2(n-1))$.
The first term $2(n+1)$ is for the two dummy sequences $w_{\$}$ and $w_{\#}$
and $4m-2(n-1)$ is for the vertex sequences $v_i$.

\runtitle{Suppose that there is an Hamiltonian Path in $G$.}
Without loss of generality, we assume that the Hamiltonian path is $v_1, v_2, \cdots, v_n$.
We place the vertex sequences in this order in the cells of the array,
with the leftmost and rightmost cells containing $w_{\$}$ and $w_{\#}$, respectively.

Consider any permutation of the edge characters such that it is a common supersequence
of all vertex sequences.
Suppose we use this sequence as the deposition sequence.
Note that each edge character appears in exactly two vertex sequences.
If there is no sharing between neighboring sequences,
each edge character incurs border length of $2$ for each of the two vertex sequences,
and the total border length would be $4m$.
Since we place the vertex sequences according to the order in an Hamiltonian path,
every edge character in this path is shared by two neighboring vertex sequences,
thus, saving a border length of $2$.
In total, we have $n-1$ edges $e_{i,i+1}$ in the Hamiltonian path,
for $1 \leq i \leq n-1$.
Therefore, we can save $2(n-1)$, implying that
the border length of the masks associated with edge characters
is $4m-2(n-1)$.
Together with the masks for the dummy characters,
the total border length is
$2(n+1) + 4m - 2(n-1)$. \enspace

\runtitle{Suppose that there is an embedding for $I$ with border length
  $2(n+1) + (4m-2(n-1))$.}
For the dummy sequences,
they have to be placed at the leftmost and rightmost cells,
otherwise, the border length incurred will be greater than $2(n+1)$.
Each edge character only appears in two vertex sequences.
So to have sharing for this character,
we can only have these two sequences being neighbors in the graph $G$.
So in order to save $2(n-1)$, we need each vertex sequence
to share one character with its neighbor in the array,
so the way the sequence
in the array should lead to an Hamiltonian path.

With the above discussion, we conclude the following theorem.

\begin{theorem}
\label{thm:1DBMP}
The $\OBMP$ is NP-Hard.
\end{theorem}

\subsection{$\BMP$ on 2D array}
\label{sec:2dbmp}
In this section, we reduce the $\BMP$ on an $1\times n$ array
to $\BMP$ on an $n \times n$ array.
This implies that $\BMP$ is NP-hard.
Consider an instance $I_1$ for $\OBMP$ where
there are $n$ sequences $s_1, s_2, \cdots, s_n$
over an alphabet $\Sigma$,
and the length of $s_i$ is $\ell_i$.
Let $\ell = \max_{1 \leq i \leq n} \ell_i$.
We construct an instance $I_2$ for $\BMP$
which contains two types of sequences,
namely, the given sequence and the dummy sequence.
The alphabet used is a superset of $\Sigma$,
precisely,
$\Sigma' = \Sigma \cup \{x_1, x_2, \cdots, x_n\} \cup \{\$\}$.
The instance $I_2$ is constructed as follows.
Let $k > \ell$ be a large integer to be determined later.

\begin{itemize}
\item Dummy sequences: we create $n^2-n$ dummy sequences each containing
  one character $\$$.
\item Given sequences: for each $s_i$,
  we create a length $k$ sequence $x_i^{k-\ell_i} \cdot s_i$.
\end{itemize}

We claim that the best way to place these $n^2$ sequences 
is to put the given sequences on the top row. 
In that case, the optimal solution for $I_1$ would give
an optimal solution for $I_2$ and vice versa.

We now prove the claim.
For each cell in the array, there are four boundaries, top, bottom, left, and right.
A sequence placed in a certain cell contributes to the overall border length
an amount of four times its length minus the sharing of characters with
its four neighbors.
Consider a sequence $s$, let us denote by $\share(s,s')$ the number of characters
that can be shared between two sequences $s$ and $s'$.
Let us also denote by $gs$, $ds$, and $b$ be a given instance, a dummy sequence,
and the outmost boundary of the array.
Then we have the following relationships.
\begin{eqnarray*}
\share(gs,gs) \leq \ell, & \quad \share(gs,ds) = 0, & \quad \share(gs,b) = k, \\
 \share(ds,ds) = 1,  & \share(ds,b) = 1
\end{eqnarray*}

If we arrange all the sequences such that the given sequences are
placed on the top row, 
we would have a sharing of $(n+2) \times \share(gs,b)
=(n+2)k$.
If any of these given sequences are not placed on the top row,
we lose a sharing of at least $k$.
No matter how the sequences are placed,
the maximum sharing apart from those with the outmost boundaries
of the array is at most $4n^2\ell$.
If we set $k$ to be large enough,
e.g., $k=4n^2\ell+1$,
then any possible internal sharing (not with outmost boundaries)
is not sufficient to compensate the loss of $k$.\footnote{%
It is possible to set a smaller value of $k$ by more careful analysis.
Yet the ultimate conclusion is still the same that we have an instance
for $\BMP$ that it is the best to have all the given sequences
placed on the top row of the array.}
We have proved the claim that
all the given sequences should be placed on the top row of the array
and the following theorem follows from Theorem~\ref{thm:1DBMP}.

\begin{theorem}
\label{thm:BMP}
The (two-dimensional) $\BMP$ is NP-hard.
\end{theorem}

\section{An $O(n^{\frac{1}{4}}\log^2{n})$ approximation algorithm for the $\BMP$}
\label{sec:approx}


In this section we present an $O(n^{\frac{1}{4}}\log^2{n})$ approximation algorithm for the $\BMP$ problem. This improves on the previous $O(\sqrt{n}\log^2{n})$ approximation.
We use the approximation algorithm for the $\PBMP$ (Corollary~\ref{thm:pbmp} in Section~\ref{sec:prelim}) as a blackbox.


First, we discuss the connection between the border length and the longest common subsequence. 
We denote the longest common subsequence between two probes $s_i$ and $s_j$, of lengths $\ell_i$ and $\ell_j$, by $LCS(s_i,s_j)$. The corresponding length is denoted by $|LCS (s_i, s_j)|$.
For any probe embedding $\epsilon$, the maximum number of common deposition nucleotides between $s_i$ and
$s_j$ is $|LCS (s_i, s_j)|$, in other words, $border_{\epsilon} (s_i, s_j) \ge \ell_i + \ell_j - 2|LCS (s_i, s_j)|$. We define $d(s_i,s_j) = \ell_i + \ell_j - 2|LCS (s_i, s_j)|$.  We also observe that this distance measure is a metric on the set of input strings.

Therefore, if we can place the probes into the array such that the sum of the distances between any adjacent cells is within a factor $c$ of the optimum (we refer to this problem as the \emph{placement problem}), then we can apply the $O(\log{n})$ approximation algorithm for the $\PBMP$ and obtain an $O(c\log{n})$ approximation for the $\BMP$. Formally, we want to find a permutation $\pi : \{1,\dots,n\} \to \{1,\dots,n\} $ such that the following quantity is minimized:

\begin{align*}
S(\pi) = & \sum_{i=1}^{n-1} d(\pi(i),\pi(i+1)) - \sum_{i=1}^{\sqrt{n-1}} d(\pi(i\sqrt{n}),\pi(i\sqrt{n}+1))\\
  & + \sum_{i=1}^{\sqrt{n}} \sum_{j=1}^{\sqrt{n-1}} d(\pi(i+(j-1)\sqrt{n}),\pi(i+j\sqrt{n}))
\end{align*}

To see why the sum is defined in this way, imagine that the probes $\pi(1), \dots, \pi(\sqrt{n})$ are placed on the first row of the array in this order, $\pi(\sqrt{n} + 1), \dots, \pi(2\sqrt{n})$ on the second row and so on.

The next proposition follows, given the previous observations.
\begin{proposition}
\label{prop:clogn-approx}
A $c$-approximation algorithm for the \emph{placement problem} implies an $O(c\log{n})$ approximation algorithm for the $\BMP$.
\end{proposition}

Since it is difficult to find in polynomial time a permutation which optimizes the value $S$ on this general metric, we first embed the metric into a tree (in fact, into a distribution of trees) with $O(\log{n})$ distortion using the algorithm of Fakcharoenphol, Rao and Talwar~\cite{FRT03} (the same algorithm used in the $\MRCT$, and implicitly $\PBMP$, approximation). This idea, together with Proposition~\ref{prop:clogn-approx}, gives us the following statement.

\begin{proposition}
\label{prop:clog2n-approx}
If we can approximate the placement problem on a tree (i.e., probes are associated to vertices in a tree and the distance between two probes is the length of the unique path between them) within a factor of $c$, then we have an $O(c\log^2{n})$ approximation to the $\BMP$.
\end{proposition}

Our approximation algorithm for the \emph{placement problem on trees} is very simple: we consider the ordering of the vertices given by an Euler tour of the tree. We then prove that this is an $O(n^\frac{1}{4})$ approximation algorithm for \emph{the placement problem on trees}. Then, by Proposition~\ref{prop:clog2n-approx} we are guaranteed to have an $O(n^\frac{1}{4} \log^2{n})$ approximation algorithm for the $\BMP$ problem.

The algorithm for the $\BMP$ problem is described formally in Algorithm~\ref{alg1}.

\begin{algorithm}[H]
\caption{\label{alg1} The $O(n^\frac{1}{4} \log^2{n})$ approximation algorithm for the $\BMP$}
\begin{algorithmic}[1]
\STATE \textbf{Input:} The strings $s_1, s_2, \dots, s_n$.
\STATE Define $d (s_i, s_j) = \ell_i + \ell_j - 2|LCS (s_i, s_j)|$
\STATE Embed the metric given by this distance and the set of input points into a tree $T$ using the algorithm from~\cite{FRT03}.
\STATE Let $\pi : \{1,2,\dots,n\} \to \{1,2,\dots,n\}$ be an Euler tour of the tree $T$.
\STATE Place the probes in the array according to $\pi$: the probes $\pi(1), \dots, \pi(\sqrt{n})$ are placed on the first row of the array in this order, $\pi(\sqrt{n} + 1), \dots, \pi(2\sqrt{n})$ on the second row and so on.
  (See Figure~\ref{fig:Euler}).
\STATE Apply the $\PBMP$ approximation algorithm.
\STATE \textbf{Output:} The placement of the probes on the array based on the Euler tour and the embedding of the probes given by the $\PBMP$ approximation algorithm.

\end{algorithmic}
\end{algorithm}
\begin{figure}[t]
  \begin{center}
  \subfigure[]{\includegraphics[height=4cm]{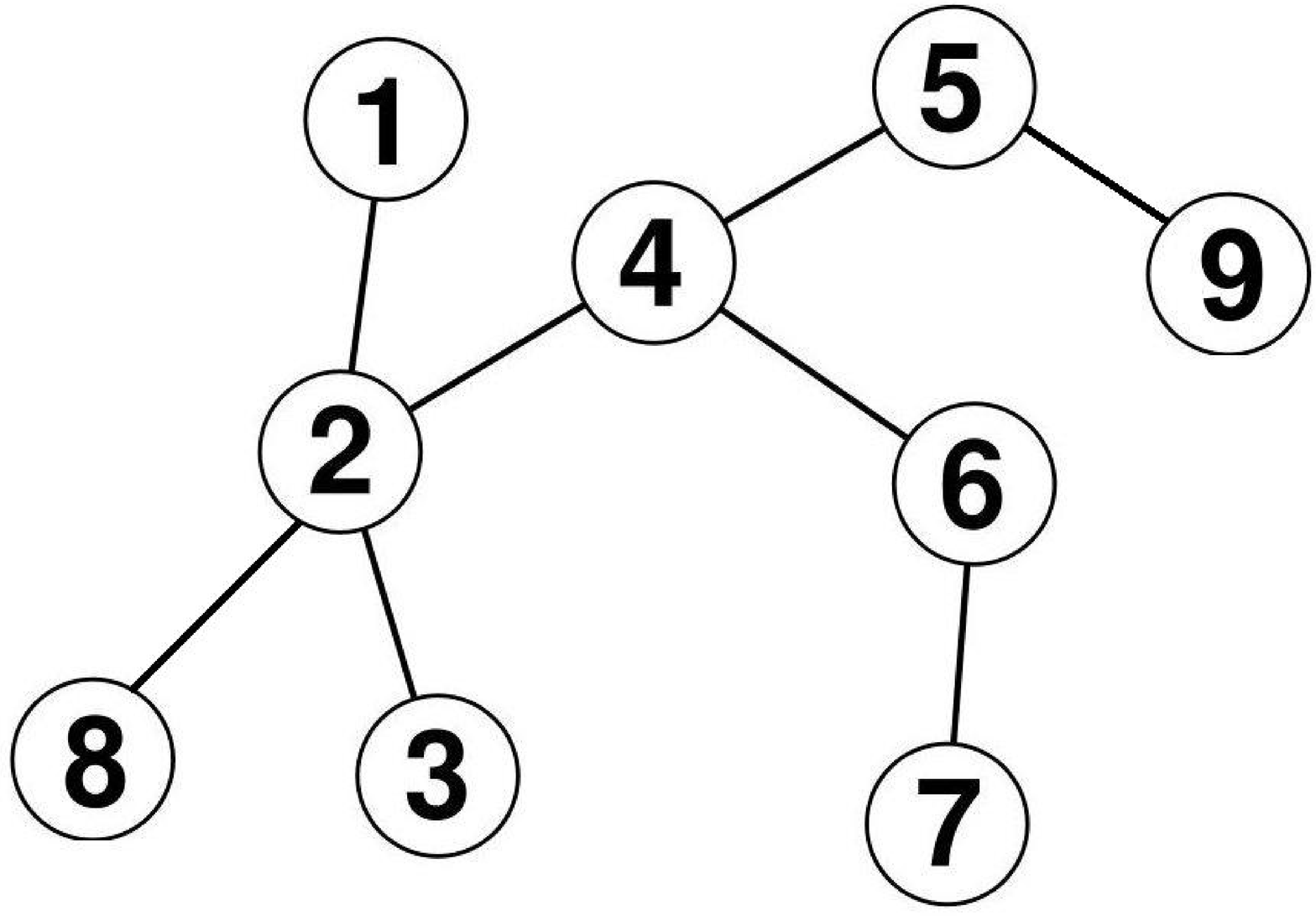}\label{fig:Euler1}}
  \hspace{1ex}
  \subfigure[]{\includegraphics[height=4cm]{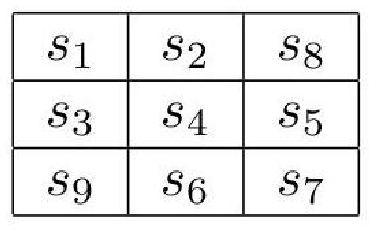}\label{fig:Euler2}}
   \caption{(a) Suppose the embedding in~\cite{FRT03} returns such a tree for $9$ probes.
   (b)~The placement of these probes on the array according to an Euler tour of the tree. }
  \label{fig:Euler}
  \end{center}
\end{figure}

Before analyzing the algorithm, let us introduce a few notations. We denote by $T$ the tree obtained after the tree embedding. Notice that the cost of a solution is given by summing each edge of $T$ several times.  We say that an edge $(x,y) \in T$ is \emph{crossed} $r$ times in a solution $\pi$ if it is used $r$ times in the solution, i.e., it belongs to exactly $r$ of the $2\sqrt{n}(\sqrt{n} - 1)$ paths of the solution.

Now, we want to find a lower bound for the optimal solution. We do so, by showing that in any solution, each edge of the tree has to be crossed at least a certain number of times. This is stated formally in the following lemma.
Let $(x,y) \in T$ and denote by $A$ and $B$ the two connected components resulted from removing $(x,y)$.

\begin{lemma}
\label{thm:cross_lower}
In any permutation $\pi$ the edge $(x,y)$ is crossed at least $\sqrt{\min(|A|,|B|)}$ times.
\end{lemma}
\begin{proof}
If we consider the probes placed on a grid graph (i.e., an $\sqrt{n} \times \sqrt{n}$ array), then the two sets of probes $A$ and $B$ determine a cut in the graph. We argue that the size of the cut is exactly the number of times the edge $(x,y)$ is crossed: for each edge $(\pi(i),\pi(j))$ in the cut, we have to add to the solution the corresponding path $\pi(i) \to \pi(j)$. But the path $\pi(i) \to \pi(j)$ has to cross the edge $(x,y)$, since $\pi(i) \in A$ and $\pi(j) \in B$.  The minimum cut determined by two sets of size $|A|$ and $|B|$ has size $\sqrt{\min(|A|,|B|)}$ and therefore the lemma follows.
\end{proof}

We give an upper bound by considering the ordering of vertices given by an Euler tour of the tree.

\begin{lemma}
\label{thm:cross_upper}
In an Euler tour ordering, $(x,y)$ is crossed at most $O(\min(\sqrt{n},|A|,|B|))$ times.
\end{lemma}
\begin{proof}
Due to the Euler tour, each edge can be crossed by edges from the paths $\pi(i) \to \pi(i+1)$ at most twice. Then we have to count how many edges from the paths $\pi(i) \to \pi(i+\sqrt{n})$ cross the edge $(x,y)$.
We argue that $(x,y)$ cannot be crossed more than $4\min(|A|,|B|)$ times. Suppose $A$ is the set with the smaller cardinality. In the worst case for each element in $A$ all its four neighbors are in $B$ and, therefore $(x,y)$ is crossed $4 \cdot \min(|A|,|B|)$ (this is actually too pessimistic but this suffices for our analysis since we are not interested in the precise constants).

We also argue that $(x,y)$ cannot be crossed more than $O(\sqrt{n})$ times. Since we follow an Euler tour, for an element $\pi(i)$, we have two cases: either $\pi(i+\sqrt{n}) \in A$, or $\pi(i+\sqrt{n}) \notin A$ and $\pi(j+\sqrt{n}) \notin A, \forall j > i$. Therefore, for only $\sqrt{n}$ elements $\pi(i)$ of $A$, $\pi(i + \sqrt{n})$ is in $B$.
The lemma then follows.
\end{proof}

\begin{theorem}
The placement of the probes in the $\sqrt{n} \times \sqrt{n}$ array in the order given by the Euler tour gives an $O(n^{\frac{1}{4}} \log^2{n})$ approximation to the $\BMP$ problem.
\end{theorem}
\begin{proof}
Consider an edge $(x,y) \in T$.

If $\min(|A|,|B|) \leq \sqrt{n}$, then $\min(\sqrt{n},|A|,|B|)/\sqrt{\min(|A|,|B|)} = \sqrt{\min(|A|,|B|)} \leq n^{\frac{1}{4}}$.


If $\min(|A|,|B|) > \sqrt{n}$, then $\min(\sqrt{n},|A|,|B|)/\sqrt{\min(|A|,|B|)} < \sqrt{n} / n^{\frac{1}{4}} = n^{\frac{1}{4}}$.

%
\noindent
We then apply Proposition~\ref{prop:clog2n-approx} and Lemmas~\ref{thm:cross_lower} and~\ref{thm:cross_upper} and the theorem follows.
\end{proof}

\section{Concluding remarks}
\label{sec:conclude}

In this paper we give a thorough study of different variations of
the \emph{Border Minimization Problem}.
We prove NP-hardness results and improved approximation algorithm.
For $\PBMP$ in which the position of the probes
is given and the goal is to find the embedding,
we show the hardness via a reduction from the
Shortest Common Supersequence problem.
For $\OBMP$
(position not given) in which the array is an 1D array, we give an
NP-hardness reduction from the Hamiltonian Path problem.
We further show that $\OBMP$ can be reduced to $\BMP$ and thus,
$\BMP$ is NP-hard.
We also give a better approximation algorithm for $\BMP$.

Contrasting with the previous result in \cite{LWY+08} that
$\OPBMP$ is polynomial time solvable,
our hardness results show that
(i) the dimension differentiates the complexity of $\PBMP$;
(ii) for 1D array, whether placement is given differentiates
the complexity of $\BMP$;
(iii) $\BMP$ is NP-hard regardless of the dimension of the array.

In the reduction from Hamiltonian Path problem to $\OBMP$,
the size of the alphabet is polynomial in terms of the number of sequences.
An interesting question is whether $\OBMP$ stays NP-hard even
for constant of characters in the alphabet.
Another natural open question is to further improve approximation
algorithms for the problem and/or to derive inapproximability
results.

  \bibliographystyle{abbrv}
  \bibliography{border}

\end{document}